\documentclass[12pt]{article}
\usepackage{amsmath}
\usepackage{graphicx,psfrag,epsf}
\usepackage{enumerate}
\usepackage{psfrag,epsf}
\usepackage{url} 
\usepackage{amsmath,amssymb,amsfonts,amsthm,mathtools,mathrsfs}

\usepackage{bm,bbm}
\usepackage{color}
\usepackage{subfigure}
\usepackage{algorithm,algpseudocode}
\usepackage[symbol]{footmisc}
\usepackage{scalefnt}
\usepackage{authblk} 
\usepackage{multirow,centernot}
\usepackage[colorlinks=true, citecolor=blue, urlcolor=blue]{hyperref}
\usepackage{natbib}
\usepackage{dsfont}
\usepackage[title]{appendix}
\input cyracc.def

\usepackage[left=1in,top=1.1in,right=0.5in,bottom=1in]{geometry}

\theoremstyle{definition}

\newtheorem{theorem}{Theorem}[section]

\newtheorem{definition}{Definition}[section]
\newtheorem{lemma}[theorem]{Lemma}
\newtheorem{proposition}[theorem]{Proposition}
\newtheorem{remark}[theorem]{Remark}

\makeatletter
\def\@seccntformat#1{\@ifundefined{#1@cntformat}%
	{\csname the#1\endcsname\quad}
	{\csname #1@cntformat\endcsname}
}
\makeatother

\newif\ifShowComments
\ShowCommentstrue
\def\strutdepth{\dp\strutbox}
\def\druk#1{\strut\vadjust{\kern-\strutdepth
        {\vtop to \strutdepth{%
                \baselineskip\strutdepth\vss
                        \llap{\hbox{#1}\quad}\null}}}}

\title{\bf
On the bias of the Gini coefficient estimator for zero-truncated Poisson distributions
}

\author{
\text{Roberto Vila}$^{1}$\thanks{Corresponding author: Roberto Vila, email: {rovig161@gmail.com}
}
\,\,\,and
\text{Helton Saulo}$^{1,2}$
\\
{\small $^{1}$ Department of Statistics, University of Brasilia, Brasilia, Brazil}\\
{\small $^{2}$ Department of Economics, Federal University of Pelotas, Pelotas, Brazil}\\
}
\begin{document}
	\maketitle 	
	\begin{abstract}
		This paper analyzes the Gini coefficient estimator for zero-truncated Poisson populations, revealing the presence of bias, and provides a mathematical expression for the bias, along with a bias-corrected estimator, which is evaluated using Monte Carlo simulation methods.
		%
	\end{abstract}
	\smallskip
	\noindent
	{\small {\bfseries Keywords.} {Zero‐truncated Poisson distribution, Gini coefficient estimator, biased estimator.}}
	\\
	{\small{\bfseries Mathematics Subject Classification (2010).} {MSC 60E05 $\cdot$ MSC 62Exx $\cdot$ MSC 62Fxx.}}
%

\section{Introduction}

The Gini coefficient is a well-established measure of inequality and dispersion, widely applied in areas such as economics, ecology, and epidemiology \citep{YAO01101999, damgaard2000, ELIAZAR2014148, Abeles2020, kharazmi2023}. In epidemiology, the Gini coefficient can be used to quantify heterogeneity in the offspring distribution of infectious diseases, making it particularly suitable for analyzing superspreading events, in which a small number of individuals are responsible for a substantial share of secondary infections; see \cite{hoehle2020}.

The zero-truncated Poisson (ZTP) distribution provides a natural probabilistic model for the number of secondary infections caused by an infected individual, given that at least one secondary infection has occurred. In this sense, the Gini coefficient for ZTP populations can be an useful tool for understanding epidemic dynamics and designing effective control strategies. Despite its usefulness, the Gini coefficient estimator for ZTP populations is prone to bias, particularly in small-sample scenarios common in contact tracing or cluster-based studies. This raises important methodological concerns regarding the reliability and correction of Gini-based heterogeneity measures in discrete, truncated settings.

In this paper, we analyze the bias of the sample Gini coefficient when the underlying population follows a ZTP distribution. Our contributions are twofold: (i) we derive a simple expression for the expected value of the Gini coefficient estimator under the ZTP model, and (ii) we propose a bias-corrected version of the estimator. The performance of the corrected estimator is assessed through Monte Carlo simulation. By improving the accuracy of inequality measurement in zero-truncated count data, our findings contribute to more reliable quantification of heterogeneity in a variety of domains, particularly this work provides tools that enhance the quantitative assessment of transmission heterogeneity, which is useful in applications involving superspreading phenomena.

The rest of this paper is organized as follows.  In Section \ref{sec:02}, we present the theoretical basis and key definitions.  In Section \ref{sec:03}, we obtain a closed-form expression for the expected value of the sample Gini coefficient estimator. In Section \ref{sec:04}, we examine an illustrative Monte Carlo simulation study.  Finally, in Section \ref{sec:05}, we provide some concluding remarks.

\section{Gini coefficient for ZTP distributions}\label{sec:02}

In this section we define the zero‐truncated Poisson distribution and some of its main characteristics, in order to then explicitly determine its Gini coefficient.

\subsection{ZTP distribution}

A discrete random variable $X$ has a zero‐truncated Poisson (ZTP) distribution \citep{Johnson2005} with parameter $\lambda>0$, denoted by $X\sim\text{ZTP}(\lambda)$, if its probability mass function  is given by
\begin{align}\label{pmf-def}
P_{\lambda}(k)
=
\mathbb{P}(X=k)
=
{1\over 1-\exp(-\lambda)} \, 
{\frac{\exp(-\lambda) \lambda^{k}}{k!}},
\quad 
k = 1, 2,\ldots.
\end{align}
It is well-known that the cumulative distribution function and
the expected value of $X\sim\text{ZTP}(\lambda)$, denoted by $F_\lambda(x)$ and $\mu=\mathbb{E}(X)$, respectively, are given by
\begin{align}\label{cdf-expect}
	F_\lambda(x)
	=
	\begin{cases}
	1
	-
	\displaystyle
	{1\over 1-\exp(-\lambda)} \, 
	{\gamma(\lfloor x\rfloor+1, \lambda)\over \Gamma(\lfloor x\rfloor+1)}, & x\geqslant 1,
	\\
	0, & x<1, 
	\end{cases}
	\quad
	\text{and}
	\quad 
\mu=\frac{\lambda}{1-\exp(-\lambda)},
\end{align}
where $\lfloor x\rfloor$ is the floor function, $\Gamma(x)$ is the (complete) gamma function and $\gamma(x,y)$ is the lower incomplete gamma function.

\subsection{Gini coefficient}
The Gini coefficient assesses income inequality via the Lorenz curve, representing cumulative income share versus population percentage. It is calculated as the ratio of the area between the Lorenz curve and the line of equality to the total area under the line. Our analysis employs the conventional mean-difference-based definition
	\begin{definition}
The Gini coefficient \citep{Gini:1936} of a random variable $X$ with finite mean $\mathbb{E}(X)$ is defined as
	\begin{align}\label{Gini-coefficient}
		G={1\over 2}\, {\mathbb{E}(\vert X_1-X_2\vert)\over\mathbb{E}(X)},
	\end{align}
	where $X_1$ and $X_2$ are independent copies of $X$.
\end{definition}


\begin{proposition}\label{prop-gini-coeff}
	The Gini coefficient for $X\sim\text{ZTP}(\lambda)$ is given by
\begin{align}\label{gini:pop}
	G
	=
	1
	-
{2\exp(-\lambda)\over 1-\exp(-\lambda)} \,
\int_0^\lambda I_0(2\sqrt{\lambda t}) \exp(-t) {\rm d}t
+
	{\exp(-2\lambda)\over 1-\exp(-\lambda)} \,
I_1(2\lambda),
\end{align}
where $I_\nu(z)=(z/2)^\nu \sum_{k=0}^\infty (z/2)^{2k}/[k!\Gamma(\nu+k+1)]$	is the modified Bessel function of the first kind of order $\nu$.
\end{proposition}
\begin{proof}
Following Proposition 2.3 in \cite{Vila-Saulo2025}, a crucial step in our analysis involves determining
\begin{align}\label{charact-gini}
G
=		
2\mathbb{P}(X<X^*)
-
1
+
\mathbb{P}(X=X^*),
\end{align}
where $X$ and $X^*$ are independent, and
$X^*$ has sized-biased
distribution, that is, its probability mass function is given by
\begin{align}\label{de-cdf-X*}
P_{\lambda}^{*}(k)
=
\mathbb{P}(X^*=k)
=
{kP_{\lambda}(k)\over \mu}
= 
{\frac{\exp(-\lambda) \lambda^{k-1}}{(k-1)!}},
\quad 
k = 1, 2,\ldots.
\end{align} 
Given the independence of the random variables $X$ and $X^*$, we can combine the results of \eqref{cdf-expect} and \eqref{de-cdf-X*} to arrive at
\begin{align}\label{pro-1}
	\mathbb{P}(X<X^*)
	=
	\sum_{k=1}^{\infty}
	F_\lambda(k-1)
	P_{\lambda}^{*}(k)
	&=
	\sum_{k=1}^{\infty}
	\left[
	1-
		{1\over 1-\exp(-\lambda)} \, 
	{\gamma(k, \lambda)\over \Gamma(k)}
	\right]
	P_{\lambda}^{*}(k)
	\nonumber
	\\[0,2cm]
&=
	1
	-
	{\exp(-\lambda) \over 1-\exp(-\lambda)}
	\sum_{k=1}^{\infty}
	{\gamma(k, \lambda)\over \Gamma(k)}
	\frac{\lambda^{k-1}}{(k-1)!}
		\nonumber
	\\[0,2cm]
	&=
	1-
	{\exp(-\lambda)\over 1-\exp(-\lambda)} \,
	\int_0^\lambda I_0(2\sqrt{\lambda t}) \exp(-t) {\rm d}t,
\end{align}
where the last equality utilizes the definitions of the lower incomplete gamma, $\gamma(s,x)=\int_0^x t^{s-1}\exp(-t){\rm d}t$, and the modified Bessel function of the first kind of order zero.
Analogously, we obtain,
\begin{align}\label{pro-2}
		\mathbb{P}(X=X^*)
	=
		\sum_{k=1}^{\infty}
	P_{\lambda}(k)
	P_{\lambda}^{*}(k)
	=
	{\exp(-2\lambda)\over 1-\exp(-\lambda)} 
	\sum_{k=1}^{\infty}
	\frac{\lambda^{2k-1}}{(k-1)!k!}
	=
	{\exp(-2\lambda)\over 1-\exp(-\lambda)} \,
	I_1(2\lambda),
\end{align}
where we have used the definition of the modified Bessel function of the first kind of order one.
By substituting equations \eqref{pro-1} and \eqref{pro-2} into formula \eqref{charact-gini}, the proof is readily obtained.
\end{proof}

\section{Bias of the Gini coefficient estimator}\label{sec:03}

The primary objective of this section is to obtain a concise closed-form representation for the expected value of the Gini coefficient estimator, $\widehat{G}$, initially introduced in the literature by \cite{Delta2003},
\begin{align}\label{gini-estimadtor-def}
	\widehat{G}
	=
	{1\over n-1} 
	\left[\dfrac{\displaystyle\sum_{1\leqslant i<j\leqslant n}\vert X_i-X_j\vert}{\displaystyle\sum_{i=1}^{n} X_i}\right],
	\quad 
	n\in\mathbb{N}, \, n\geqslant 2,
\end{align}
where $X_1, X_2,\ldots,X_n$ are independent, identically distributed (i.i.d.) observations from the population.

In the following, Theorem 3.1 from \cite{Vila-Saulo2025} is reformulated for zero-truncated Poisson distributions, enabling the calculation of the Gini coefficient estimator bias (see Subsection \ref{The Gini coefficient estimator bias}).

\begin{theorem}\citep{Vila-Saulo2025}
	\label{the-main-2}
	Let $X_1,X_2,\ldots$ be independent copies of  $X\sim\text{ZTP}(\lambda)$ and common cumulative distribution function $F_{\lambda}(x)$ and probability mass function $P_{\lambda}(x)$, given in \eqref{cdf-expect} and \eqref{pmf-def}, respectively.
	The following holds: 
	\begin{align*}
		\mathbb{E}(\widehat{G})
		=
		n
		\mathbb{E}(X)
		\left\{2R_{1}(F)-R_{\infty}(F)		+
		\mathbb{E}
		\left[g(X^*,X)\mathds{1}_{\{X=X^*\}}\right]\right\},
		\quad 
		n\in\mathbb{N}, \, n\geqslant 2,
	\end{align*}
	where $X^*$ has sized-biased distribution \eqref{de-cdf-X*}, 
	$R_{1}(F)\equiv \lim_{\varepsilon\to 1}R_{\varepsilon}(F)$, 
	$R_{\infty}(F)\equiv \lim_{\varepsilon\to \infty}R_{\varepsilon}(F)$,
	\begin{align}\label{def-R-function}
		&R_{\varepsilon}(F)
		\equiv
		\int_{0}^{\infty}
		\mathbb{E}\left[
		\exp\left(-xX^*\right)
		H(x,\varepsilon X^*)
		\right]
		\mathscr{L}^{n-2}_F(x){\rm d}x,
		\quad 
		\varepsilon>0,
		\\[0,2cm]
		&
		H(x,x^*)
		=
			\sum_{k=1}^{x^*-1}
		\exp(-xk)
		P_{\lambda}(k),
		\quad 	 x>0, \ x^*=1,2,\ldots,
		\label{def-H-function}
		\\[0,2cm]
		&
		{g}(u,v)
		\equiv 
		\int_{0}^{\infty}\exp\left[-(u+v) x\right] \mathscr{L}^{n-2}_F(x){\rm d}x,
		\quad u,v>0,
		\label{def-g}
	\end{align}
	and
	$\mathscr{L}_F(y)
	=\sum_{k=1}^\infty \exp(-yk)P_{\lambda}(k)$ is the Laplace transform corresponding to distribution $F_{\lambda}(x)$.
	In the above, we are assuming that the improper integrals involved exist.
\end{theorem}

\subsection{Technical results}

As the determination of $\mathbb{E}(\widehat{G})$ relies on the explicit evaluation of $R_1(F)$, $R_\infty(F)$, and $\mathbb{E}
\left[g(X^*,X)\mathds{1}_{\{X=X^*\}}\right]$ (see Theorem \ref{the-main-2}), this subsection presents several technical results that facilitate the computation of these quantities, as demonstrated in Theorems \ref{pro-main-1} and \ref{pro-main-2}.

\begin{proposition}\label{prop-1}
The following holds:
\begin{align*}
		H(x,x^*)
		=
	{\exp(-\lambda) \over 1-\exp(-\lambda)} 
\left\{
{\exp[\lambda\exp(-x)] \Gamma(x^*, \lambda\exp(-x))\over\Gamma(x^*)}-1
\right\},
		\quad 
	 x>0, \ x^*=1,2,\ldots,
\end{align*}
with $H(x,x^*)$ being as in \eqref{def-H-function}.
\end{proposition}
\begin{proof}
Direct calculation yields
\begin{align*}
	H(x,x^*)
	=
	\sum_{k=1}^{x^*-1}
	\exp(-xk)
	P_\lambda(k)
	&=
		{\exp(-\lambda) \over 1-\exp(-\lambda)} 
	\sum_{k=1}^{x^*-1}
	{\frac{[\lambda\exp(-x)]^{k}}{k!}}
	\nonumber
	\\[0,2cm]
	&
	=
	{\exp(-\lambda) \over 1-\exp(-\lambda)} 
	\left\{
	{\exp[\lambda\exp(-x)] \Gamma(x^*, \lambda\exp(-x))\over\Gamma(x^*)}-1
	\right\},
	\end{align*}
	with the last equality following from the identity  $\sum_{k=1}^{n-1} a^k/k!=\exp(a) \Gamma(n, a)/\Gamma(n)-1$.
This completes the proof.
\end{proof}

\begin{proposition}\label{LT}
We have:
\begin{align*}
	\mathscr{L}_F(x)
	=
	{\exp(-\lambda) \over 1-\exp(-\lambda)} 
\left\{
\exp[\lambda\exp(-x)]-1
\right\},
	\quad 
	x>0,
\end{align*}
where $\mathscr{L}_F(p)$ is the Laplace transform corresponding to distribution $F_\lambda(x)$ given in \eqref{cdf-expect}.
\end{proposition}
\begin{proof}
%
The proof is a direct consequence of Proposition \ref{prop-1}, as we note that the equality $	\mathscr{L}_F(x)=\lim_{\varepsilon\to\infty}	H(x,\varepsilon x^*)$, for $x^*=1,2,\ldots$, holds.
\end{proof}

\begin{proposition}\label{prop-exp}
The following statement is valid:
\begin{align*}
\mathbb{E}&[\exp(-xX^*) H(x,\varepsilon X^*)]
\\[0,2cm]
	&=
	{\exp(-2\lambda)\over 1-\exp(-\lambda)} \,
\exp[\lambda\exp(-x)-x] 
\left\{
\int_{\lambda\exp(-x)}^\infty
\left\{
\sum_{k=0}^{\infty}
{
	t^{\varepsilon (k+1)-1}
	\over\Gamma(\varepsilon (k+1))} \,
{\frac{[\lambda\exp(-x)]^{k}}{k!}} 
\right\}
\exp(-t){\rm d}t
-
1
\right\},
\end{align*}
with $H(x,x^*)$, $x>0$, $x^*=1,2,\ldots$, being as in Proposition \ref{prop-1}.
\end{proposition}
\begin{proof}
By using Proposition \ref{prop-1} and the definition  \eqref{de-cdf-X*} of the  probability mass function of $X^*$, we have
\begin{align}
	&\mathbb{E}[\exp(-xX^*) H(x,\varepsilon X^*)]
	=
	\sum_{k=1}^{\infty}
	\exp(-kx) H(x,\varepsilon k) 
P_{\lambda}^{*}(k) \nonumber
		\\[0,2cm]
	&=
			{\exp(-2\lambda) \over 1-\exp(-\lambda)} 
		\sum_{k=1}^{\infty}
	\exp(-kx) 
	\left\{
	{\exp[\lambda\exp(-x)] \Gamma(\varepsilon k, \lambda\exp(-x))\over\Gamma(\varepsilon k)}-1
	\right\}
	{\frac{\lambda^{k-1}}{(k-1)!}} \nonumber
	\\[0,2cm]
	&=
	{\exp(-2\lambda)\over 1-\exp(-\lambda)} \,
	\exp[\lambda\exp(-x)-x] 
	\left\{
	\sum_{k=1}^{\infty}
	{\Gamma(\varepsilon k, \lambda\exp(-x))\over\Gamma(\varepsilon k)} \,
	{\frac{[\lambda\exp(-x)]^{k-1}}{(k-1)!}}
	-1
	\right\}, \label{exp-int}
\end{align} 
where  in the last step, we applied a change of variable and utilized the fact that $\sum_{k=0}^{\infty} \exp[-(k+1)x]\lambda^k/k!=\exp[\lambda\exp(-x)-x]$.
Upon invoking the definition $\Gamma(s,x)=\int_x^\infty t^{s-1}\exp(-t){\rm d}t$ of upper incomplete gamma function, the above expression can be written as
\begin{align*}
		&=
	{\exp(-2\lambda)\over 1-\exp(-\lambda)} \,
	\exp[\lambda\exp(-x)-x] 
	\left\{
	\int_{\lambda\exp(-x)}^\infty
	\left\{
		\sum_{k=0}^{\infty}
	{
	 t^{\varepsilon (k+1)-1}
	\over\Gamma(\varepsilon (k+1))} \,
	{\frac{[\lambda\exp(-x)]^{k}}{k!}} 
	\right\}
	\exp(-t){\rm d}t
	-
	1
	\right\}.
\end{align*}

The proof is now finished.
\end{proof}

\begin{lemma}\label{lemma-main-1}
It holds that:
\begin{enumerate}
\item 
\begin{align*}
&\lim_{\varepsilon\to 1}
\mathbb{E}[\exp(-xX^*) H(x,\varepsilon X^*)]
\\[0,2cm]
&=
		{\exp(-2\lambda)\over 1-\exp(-\lambda)} \,
\exp[\lambda\exp(-x)-x] 
\left\{
\exp[\lambda\exp(-x)]
Q_1\left(\sqrt{2\lambda\exp(-x)},\sqrt{2\lambda\exp(-x)}\right)
-1
\right\}.
\end{align*}
\item 
\begin{align*}
	\lim_{\varepsilon\to\infty}
	\mathbb{E}[\exp(-xX^*) H(x,\varepsilon X^*)]
	=
	{\exp(-2\lambda)\over 1-\exp(-\lambda)} \,
	\exp[\lambda\exp(-x)-x] 
	\left\{
	\exp[\lambda\exp(-x)] 
	-1
	\right\}.
\end{align*}
\end{enumerate}
In the above, $Q_{\nu }(a,b)=
a^{1-\nu}\int _{b}^{\infty }x^{\nu }\exp \left[-(x^{2}+a^{2})/{2}\right]I_{\nu -1}(ax){\rm d}x
$ denotes the generalized Marcum Q-function of order $\nu$.
\end{lemma}
\begin{proof}
	Taking the limit as $\varepsilon$ approaches 1 in Proposition \ref{prop-exp}, we obtain
	\begin{align*}
	\lim_{\varepsilon\to 1}
	\mathbb{E}[\exp(-x&X^*) H(x,\varepsilon X^*)]
	\\[0,1cm]
	&=
		{\exp(-2\lambda)\over 1-\exp(-\lambda)} \,
	\exp[\lambda\exp(-x)-x] 
	\left\{
	\int_{\lambda\exp(-x)}^\infty
	\left\{
	\sum_{k=0}^{\infty}
	{\frac{[\lambda\exp(-x)t]^{k}}{(k!)^2}} 
	\right\}
	\exp(-t){\rm d}t
	-
	1
	\right\}
		\\[0,2cm]
	&=
	{\exp(-2\lambda)\over 1-\exp(-\lambda)} \,
	\exp[\lambda\exp(-x)-x] 
	\left\{
	\int_{\lambda\exp(-x)}^\infty
	I_0\big(2\sqrt{\lambda\exp(-x)t}\, \big)
	\exp(-t){\rm d}t
	-
	1
	\right\},
	\end{align*}
	where the last equality utilizes the definition of the  modified Bessel function of the first kind of order zero.
Taking the change of variable $s=\sqrt{2t}$, the above expression becomes
	\begin{align*}
		&=
			{\exp(-2\lambda)\over 1-\exp(-\lambda)} \,
		\exp[\lambda\exp(-x)-x] 
		\left\{
		\int_{\sqrt{2\lambda\exp(-x)}}^\infty
		I_0\big(\sqrt{2\lambda\exp(-x)}\, s\big)
		\exp(-s^2/2)
		s{\rm d}s
		-
		1
		\right\}.
	\end{align*}
	By using the definition of generalized Marcum Q-function of order one, the proof of statement in Item 1 follows.

Conversely, taking the limit as $\varepsilon$ approaches infinity in \eqref{exp-int} and using the fact that $\lim_{x\to\infty} \Gamma(x, y)/\Gamma(x) = 1$, we derive
\begin{align*}
	\lim_{\varepsilon\to\infty}
\mathbb{E}[\exp(-xX^*) H(x,\varepsilon X^*)]
=
	{\exp(-2\lambda)\over 1-\exp(-\lambda)} \,
\exp[\lambda\exp(-x)-x] 
\left\{
\sum_{k=1}^{\infty}
{\frac{[\lambda\exp(-x)]^{k-1}}{(k-1)!}}
-1
\right\},
\end{align*}
from which the proof of Item 2 follows.

	This concludes the proof of lemma.
\end{proof}

Having established Lemma \ref{lemma-main-1}, we can now proceed to state and prove the following key theorem.
\begin{theorem}\label{pro-main-1}
The functions $R_1(F)$ and $R_\infty(F)$, defined in Theorem \ref{the-main-2},  satisfy the following identities:
\begin{enumerate}
\item 
\begin{align*}
R_1(F)
=
{\exp(-n\lambda) \over 2[1-\exp(-\lambda)]^{n-1}} 
\int_{0}^{1}
I_0(2\lambda y)
\left[
\exp(\lambda y)-1
\right]^{n-2}
{\rm d}y
+
{\exp(-\lambda)+n-1\over 2n(n-1)\lambda}
-
{\exp(-\lambda)\over (n-1)\lambda}.
\end{align*}
\item 
\begin{align*}
R_\infty(F)
=
{1-\exp(-\lambda)\over n\lambda}.
\end{align*}
\end{enumerate}
In the above, $I_0(z)$	is the modified Bessel function of the first kind of order zero.
\end{theorem}
\begin{proof}
Utilizing the definition \eqref{def-R-function} of $R_1(F)$ and the result of Lemma \ref{lemma-main-1}, we have
\begin{align*}
R_{1}(F)
&=
\int_{0}^{\infty}
\lim_{\varepsilon\to 1}
\mathbb{E}\left[
\exp\left(-xX^*\right)
H(x,\varepsilon X^*)
\right]
\mathscr{L}^{n-2}_F(x){\rm d}x
\\[0,2cm]
&=
{\exp(-n\lambda) \over [1-\exp(-\lambda)]^{n-1}} 
\int_{0}^{\infty}
\exp[\lambda\exp(-x)-x] 
\\[0,2cm]
&\times 
\left\{
\exp[\lambda\exp(-x)]
Q_1\left(\sqrt{2\lambda\exp(-x)},\sqrt{2\lambda\exp(-x)}\right)
-1
\right\}
\left\{
\exp[\lambda\exp(-x)]-1
\right\}^{n-2}
{\rm d}x
\\[0,2cm]
&=
{\exp(-n\lambda) \over [1-\exp(-\lambda)]^{n-1}} 
\int_{0}^{1}
\exp(2\lambda y)
Q_1\left(\sqrt{2\lambda y},\sqrt{2\lambda y}\right)
\left[
\exp(\lambda y)-1
\right]^{n-2}
{\rm d}y
-
{\exp(-\lambda)\over (n-1)\lambda},
\end{align*}
where, in the last line, the change of variable $y = \exp(-x)$ was made.
By using the identity \citep[see][]{Brychkov2012}: $Q_1(a,a)=[\exp(-a^2)I_0(a^2)+1]/2$, the last expression becomes
\begin{align*}
&=
{\exp(-n\lambda) \over 2[1-\exp(-\lambda)]^{n-1}} 
\int_{0}^{1}
I_0(2\lambda y)
\left[
\exp(\lambda y)-1
\right]^{n-2}
{\rm d}y
\\[0,2cm]
&+
{\exp(-n\lambda) \over 2[1-\exp(-\lambda)]^{n-1}} 
\int_{0}^{1}
\exp(2\lambda y)
\left[
\exp(\lambda y)-1
\right]^{n-2}
{\rm d}y
-
{\exp(-\lambda)\over (n-1)\lambda}
\\[0,2cm]
&=
{\exp(-n\lambda) \over 2[1-\exp(-\lambda)]^{n-1}} 
\int_{0}^{1}
I_0(2\lambda y)
\left[
\exp(\lambda y)-1
\right]^{n-2}
{\rm d}y
\\[0,2cm]
&+
{\exp(-n\lambda) \over 2[1-\exp(-\lambda)]^{n-1}}  \,
{[\exp(\lambda)-1]^{n-1}[(n-1)\exp(\lambda)+1]\over n(n-1)\lambda}
-
{\exp(-\lambda)\over (n-1)\lambda}.
\end{align*}
This concludes the proof of Item 1.

On the other hand, using the definition of $R_\infty(F)$ (see Theorem \ref{the-main-2}) and the result of Lemma \ref{lemma-main-1}, we obtain
\begin{align*}
	R_{\infty}(F)
&=
	\int_{0}^{\infty}
	\lim_{\varepsilon\to \infty}
	\mathbb{E}\left[
	\exp\left(-xX^*\right)
	H(x,\varepsilon X^*)
	\right]
	\mathscr{L}^{n-2}_F(x){\rm d}x
	\\[0,2cm]
	&=
	{\exp(-n\lambda) \over [1-\exp(-\lambda)]^{n-1}} 
	\int_{0}^{\infty}
	\exp[\lambda\exp(-x)-x] 
	\left\{
	\exp[\lambda\exp(-x)]-1
	\right\}^{n-1}
	{\rm d}x
		\\[0,2cm]
	&=
	{\exp(-n\lambda) \over [1-\exp(-\lambda)]^{n-1}} 
	\,
	{[\exp(\lambda)-1]^n\over n\lambda}.
\end{align*}
Then the proof of Item 2 is readily obtained.
\end{proof}

\begin{theorem}\label{pro-main-2}
		The following holds:
\begin{align*}
	\mathbb{E}
	\left[g(X^*,X)\mathds{1}_{\{X=X^*\}}\right]
	=
		{\exp(-n\lambda) \over [1-\exp(-\lambda)]^{n-1}} 
\int_{0}^{1}
I_1(2\lambda y)
\left[
\exp(\lambda y)-1
\right]^{n-2}
{\rm d}y,
\end{align*}
where ${g}(u,v)$ is as given in \eqref{def-g} and $I_1(z)$	is the modified Bessel function of the first kind of order one.
\end{theorem}
\begin{proof}
By using Proposition \ref{LT} and then by introducing the change of variable $y=\exp(-x)$, we obtain an alternative representation of $g(k,k)$ in \eqref{def-g}  as
\begin{align*}
	{g}(k,k)
=
	\int_{0}^{\infty}
	\exp(-2kx) \mathscr{L}^{n-2}_F(x)
	{\rm d}x
	&=
	{\exp[-(n-2)\lambda] \over [1-\exp(-\lambda)]^{n-2}} 
	\int_{0}^{\infty}
\exp(-2kx) 
\left\{
\exp[\lambda\exp(-x)]-1
\right\}^{n-2}
{\rm d}x
		\\[0,2cm]
	&=
	{\exp[-(n-2)\lambda] \over [1-\exp(-\lambda)]^{n-2}} 
	\int_{0}^{1}
y^{2k-1} 
\left[
\exp(\lambda y)-1
\right]^{n-2}
{\rm d}y.
\end{align*}
Consequently, by using the above representation of $g(k,k)$, we observe that
%
\begin{align*}
	\mathbb{E}
	\left[g(X^*,X)\mathds{1}_{\{X=X^*\}}\right]
	&=
	\sum_{k=1}^{\infty}
	g(k,k)
	P_\lambda(k)
	P_\lambda^*(k)
	\\[0,2cm]
	&=
	{\exp(-n\lambda) \over [1-\exp(-\lambda)]^{n-1}} 
\int_{0}^{1}
\left[
\exp(\lambda y)-1
\right]^{n-2}
\left[
	\sum_{k=1}^{\infty}
{\frac{(\lambda y)^{2k-1}}{(k-1)!k!}}
\right]
{\rm d}y
		\\[0,2cm]
	&=
		{\exp(-n\lambda) \over [1-\exp(-\lambda)]^{n-1}} 
	\int_{0}^{1}
	I_1(2\lambda y)
	\left[
	\exp(\lambda y)-1
	\right]^{n-2}
	{\rm d}y,
\end{align*}
where the last step uses the definition of the modified Bessel function of the first kind of order one.

This concludes the proof.
\end{proof}

\subsection{The Gini coefficient estimator bias} \label{The Gini coefficient estimator bias}

As $\mu=\mathbb{E}(X)={\lambda}/{[1-\exp(-\lambda)]}$ for $X\sim\text{ZTP}(\lambda)$, by applying Theorem \ref{the-main-2} of \cite{Vila-Saulo2025}, we have
\begin{align*}
		\mathbb{E}(\widehat{G})
&=
n
\mu
\left\{
2R_{1}(F)
-
R_{\infty}(F)		
+
\mathbb{E}
\left[g(X^*,X)\mathds{1}_{\{X=X^*\}}\right]
\right\},
\end{align*}
where $R_{1}(F), R_{\infty}(F)$ and $g(u,v)$ are defined in Theorem \ref{the-main-2}.
Hence, from Theorems \ref{pro-main-1} and \ref{pro-main-2} it follows that
%
%
%
%
%
%
%
\begin{align}\label{expect-Gini}
	\mathbb{E}(\widehat{G})
	=
	{n\lambda\exp(-n\lambda) \over [1-\exp(-\lambda)]^{n}} 
\int_{0}^{1}
[I_0(2\lambda y)+I_1(2\lambda y)]
\left[
\exp(\lambda y)-1
\right]^{n-2}
{\rm d}y
-
{n\exp(-\lambda)\over (n-1)[1-\exp(-\lambda)]}.
\end{align}

Therefore, by \eqref{expect-Gini} and Proposition \ref{prop-gini-coeff}, the bias of $\widehat{G}$ relative to $G$, denoted by $\text{Bias}(\widehat{G},G)$, can be written
as
\begin{align}\label{bias-gini-geometric}
	\text{Bias}(\widehat{G},G)
	&=
	{n\lambda\exp(-n\lambda) \over [1-\exp(-\lambda)]^{n}} 
\int_{0}^{1}
[I_0(2\lambda y)+I_1(2\lambda y)]
\left[
\exp(\lambda y)-1
\right]^{n-2}
{\rm d}y
-
{n\exp(-\lambda)\over (n-1)[1-\exp(-\lambda)]}
\nonumber
\\[0,2cm]
&
+
{2\exp(-\lambda)\over 1-\exp(-\lambda)} \,
\int_0^\lambda I_0(2\sqrt{\lambda t}) \exp(-t) {\rm d}t
-
{\exp(-2\lambda)\over 1-\exp(-\lambda)} \,
I_1(2\lambda)
-
1.
\end{align}

\begin{remark}
It is worth noting that the integrals in \eqref{expect-Gini} and \eqref{bias-gini-geometric} cannot be expressed in terms of standard mathematical functions, but they can be easily evaluated numerically.
\end{remark}

\section{Illustrative simulation study}\label{sec:04}

Note that a bias-corrected Gini estimator can then be proposed from \eqref{gini-estimadtor-def} and \eqref{bias-gini-geometric} as
\begin{align}\label{Gini-GM-BC} \nonumber
\widehat{G}_{\text{bc}} &= \widehat{G} - \text{Bias}(\widehat{G}, G)
\\[0,2cm] 
\nonumber
&=	{1\over n-1}
	\left[\dfrac{\displaystyle\sum_{1\leqslant i<j\leqslant n}\vert X_i-X_j\vert}{\displaystyle\sum_{i=1}^{n} X_i}\right]
	\\[0,2cm] \nonumber
&-\left\{
	{n\widehat{\lambda}\exp(-n\widehat{\lambda}) \over [1-\exp(-\widehat{\lambda})]^{n}}
\int_{0}^{1}
[I_0(2\widehat{\lambda} y)+I_1(2\widehat{\lambda} y)]
\left[
\exp(\widehat{\lambda} y)-1
\right]^{n-2}
{\rm d}y
-
{n\exp(-\widehat{\lambda})\over (n-1)[1-\exp(-\widehat{\lambda})]} \right.
\nonumber
\\[0,2cm]
&
\left.+
{2\exp(-\widehat{\lambda})\over 1-\exp(-\widehat{\lambda})} \,
\int_0^{\widehat{\lambda}} I_0(2\sqrt{\widehat{\lambda} t}) \exp(-t) {\rm d}t
-
{\exp(-2\widehat{\lambda})\over 1-\exp(-\widehat{\lambda})} \,
I_1(2\widehat{\lambda})
-
1\right\},
	\quad
	n\in\mathbb{N}, \, n\geqslant 2,
\end{align}
where hat notation on $\lambda$ denotes the maximum likelihood estimator. We evaluate the performance of the standard \eqref{gini-estimadtor-def} and bias-corrected \eqref{Gini-GM-BC} Gini coefficient estimators in the context of ZTP populations through a Monte Carlo simulation study. The study considers four values of the ZTP parameter $\lambda \in \{0.1, 0.5, 1.0, 2.0\}$ and four different sample sizes $n \in \{5, 10, 30, 50\}$. For each configuration, we perform $1,000$ replications. In each simulation, data is generated from a ZTP($\lambda$) distribution, the parameter $\lambda$ is estimated via maximum likelihood, and both the standard and bias-corrected Gini coefficient estimators are computed. The steps of the Monte Carlo simulation study are described in Algorithm 1.

\begin{algorithm}[!ht]\label{algo}
\caption{Monte Carlo simulation for bias-corrected Gini estimator under the ZTP model.}
\label{algorithm:ztp}
\begin{algorithmic}[1]
\State \textbf{Input:} Number of simulations $N_{\text{sim}} = 1000$; sample sizes $n \in \{5, 10, 30, 50\}$;
ZTP parameter $\lambda \in \{0.1, 0.5, 1.0, 2.0\}$.
\State \textbf{Output:} Estimates of the Gini coefficient $\widehat{G}$ and its bias-corrected version $\widehat{G}_{\text{bc}}$.

\For{each $\lambda \in \{0.1, 0.5, 1.0, 2.0\}$}
  \State Compute the true Gini coefficient $G$ for ZTP($\lambda$) using the formula~\eqref{gini:pop}.
  \For{each sample size $n \in \{5, 10, 30, 50\}$}
    \For{each simulation run $s = 1, \dots, N_{\text{sim}}$}
        \State \textbf{Step 1: Generate data}
        \State Simulate $X_1, \dots, X_n \sim \text{ZTP}(\lambda)$.
        \State Estimate $\widehat{\lambda}$ via maximum likelihood.
        \State \textbf{Step 2: Compute Gini estimates}
        \State Compute empirical Gini coefficient $\widehat{G}$ using~\eqref{gini-estimadtor-def}.
        \State Compute the bias-corrected estimator using~\eqref{Gini-GM-BC}, that is,
        \[
        \widehat{G}_{\text{bc}} = \widehat{G} - \text{Bias}(\widehat{G}, G).
        \]
    \EndFor
    \State \textbf{Step 3: Compute Monte Carlo summaries}
   \State Compute relative bias and mean squared error (MSE) for $\widehat{G}$ (similarly for $\widehat{G}_{\text{bc}}$) using:
\[
\widehat{\text{Relative Bias}} = \left| \frac{\frac{1}{N_{\text{sim}}} \sum_{k=1}^{N_{\text{sim}}} \widehat{G}^{(k)} - G}{G} \right|, \quad
\widehat{\text{MSE}} = \frac{1}{N_{\text{sim}}} \sum_{k=1}^{N_{\text{sim}}} \big[\widehat{G}^{(k)} - G \big]^2,
\]
where $N_{\text{sim}}$ is the number of Monte Carlo replications and $G$ is the true Gini coefficient given in ~\eqref{gini:pop}.

  \EndFor
\EndFor
\State \textbf{Return:} Summary tables with relative bias and MSE for all configurations.
\end{algorithmic}
\end{algorithm}

Figures~\ref{fig:bias_n}-\ref{fig:mse_n} present the Monte Carlo simulation results of the standard and bias-corrected Gini estimators across different combinations of the ZTP parameter $\lambda$ and sample sizes $n$. Figure~\ref{fig:bias_n} presents the relative bias for both estimators. From this figure, we observe that the standard estimator $\widehat{G}$ exhibits relative bias that is particularly pronounced for small sample sizes and smaller values of $\lambda$. In contrast, the bias-corrected estimator $\widehat{G}_{\text{bc}}$ significantly reduces the relative bias in all scenarios, thereby demonstrating the effectiveness of the proposed correction. Figure~\ref{fig:mse_n} presents the mean squared error (MSE) for both estimators. In general, both estimators exhibit similar MSE values across the configurations considered.

\begin{figure}[!ht]
  \centering
  \includegraphics[width=0.75\textwidth]{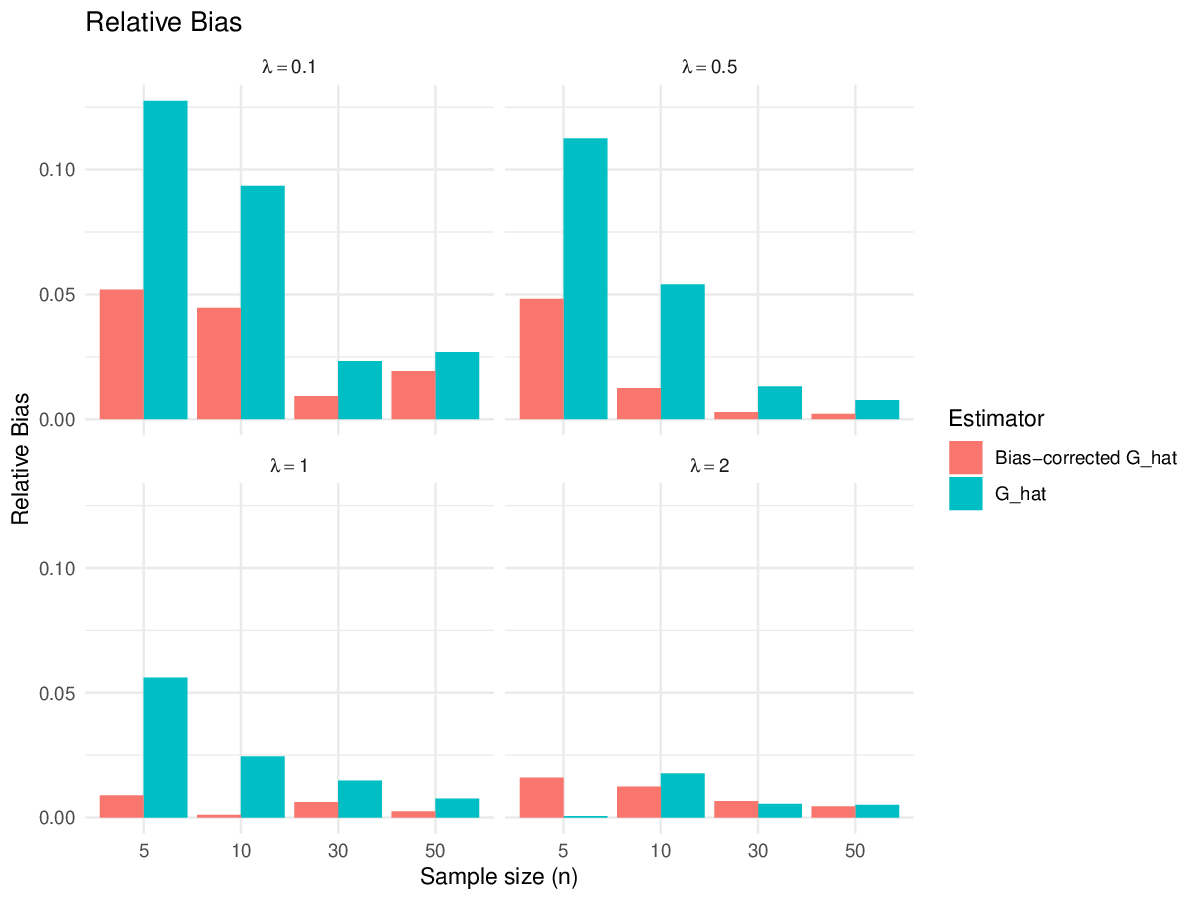}
  \caption{Relative bias of Gini estimators vs. sample size $n$, for different values of $\lambda$.}
  \label{fig:bias_n}
\end{figure}

\begin{figure}[!ht]
  \centering
  \includegraphics[width=0.75\textwidth]{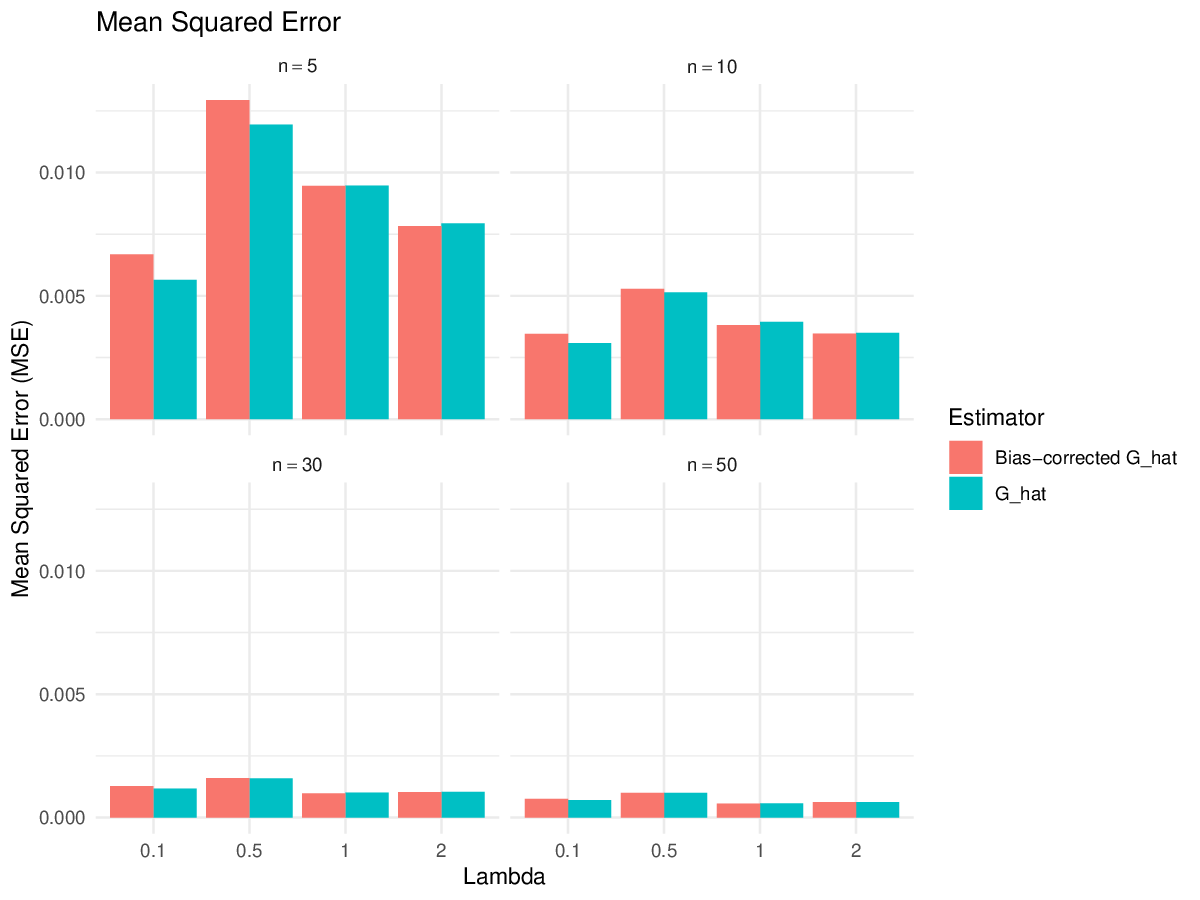}
  \caption{Mean squared error of Gini estimators vs. sample size $n$, for different values of $\lambda$.}
  \label{fig:mse_n}
\end{figure}


\section{Concluding remarks}\label{sec:05}

In this paper, we analyzed the statistical properties of the Gini coefficient estimator under zero-truncated Poisson populations. Although the Gini coefficient is a widely used inequality measure, its sample estimator exhibits bias in discrete and truncated settings such as the zero-truncated Poisson distribution. We derived a simple expression for the  bias and proposed a bias-corrected estimator that accounts for both the truncation and the finite-sample behavior. The effectiveness of the proposed correction was demonstrated through a Monte Carlo simulation study, covering a wide range of parameter values and sample sizes. The results showed that the proposed bias-corrected estimator $\widehat{G}_{\text{bc}}$ significantly reduces this bias. These findings are particularly relevant in fields such as epidemiology, where the Gini coefficient for zero-truncated Poisson populations can be used to quantify heterogeneity in the offspring distribution of infectious diseases. As part of future work, we may explore bias correction methods for other truncated or discrete distributions, as well as extensions to multivariate Gini measures and regression frameworks involving inequality indices. Work on these problems is currently in progress and we hope to report these findings in future.

\clearpage

\paragraph*{Acknowledgements}
The research was supported in part by CNPq and CAPES grants from the Brazilian government.

\paragraph*{Disclosure statement}
There are no conflicts of interest to disclose.



\end{document}